\documentclass[journal,a4paper,twocolumn]{article}

\usepackage{algorithm, algorithmic}
\usepackage{graphicx} 
\usepackage{amsmath} 
\usepackage{amssymb}  
\usepackage{amsthm}
\usepackage{amsfonts}
\usepackage{dsfont}
\usepackage{color}

\usepackage{times,amsmath,epsfig}
\usepackage{cite}
\usepackage{tikz}
\usepackage{xcolor}
\usepackage[english]{babel}
\usepackage[latin1]{inputenc}
\usepackage{subfigure}
\usepackage{hhline}
\usepackage{mathtools}

\newcommand{\R}{\mathds{R}}

\newcommand{\V}{\mathcal{V}}
\newcommand{\E}{\mathcal{E}}
\newcommand{\G}{\mathcal{G}}
\newcommand{\ROT}{\mathrm{R}}

\newcommand{\Q}{\mathrm{Q}}
\newcommand{\coeff}{\mathbf{c}}

\newtheorem{definition}{Definition}
\newtheorem{proposition}{Proposition}
\newtheorem{theorem}{Theorem}

\usepackage{etoolbox}
\newtoggle{showrevisions}

\togglefalse{showrevisions}

\iftoggle{showrevisions}{
\usepackage{xcolor}
\usepackage{changebar}
\usepackage[normalem]{ulem}

\newcommand{\removed}[1]{\cbstart\removedfragile{#1}\cbend{}}
\newcommand{\removedfragile}[1]{{\color{red}{\sout{#1}}}{}}
\newcommand{\added}[1]{\cbstart\addedfragile{#1}\cbend{}}

\newcommand{\changed}[2]{\added{#1}\removed{#2}}
}{
  \newcommand{\removed}[1]{} 
  \newcommand{\removedfragile}[1]{}
  \newcommand{\added}[1]{#1}
  
  \newcommand{\changed}[2]{#1}
}
\title{Steerable Discrete Cosine Transform}

\author{
{Giulia Fracastoro{\small $~^{\#}$}, Sophie M. Fosson{\small $~^{\#}$}, Enrico Magli{\small $~^{\#}$}}
\vspace{1.6mm}\\
\fontsize{10}{10}\selectfont\itshape
$^{\#}$\,Department of Electronics and Telecommunications, Politecnico di Torino, Italy\\
\fontsize{9}{9}\selectfont\ttfamily\upshape}

\begin{document}
\maketitle

\begin{abstract}
In image compression, classical block-based separable transforms tend to be inefficient when image blocks contain arbitrarily shaped discontinuities. For this reason, transforms incorporating directional information are an appealing alternative. In this paper, we propose a new approach to this problem, namely a discrete cosine transform (DCT) that can be steered in any chosen direction. Such transform, called steerable DCT (SDCT), allows to rotate in a flexible way pairs of basis vectors, and enables precise matching of directionality in each image block, achieving improved coding efficiency. The optimal rotation angles for SDCT can be represented as solution of a suitable rate-distortion (RD) problem. We propose iterative methods to search such solution, and we develop a fully fledged image encoder to practically compare our  techniques with other competing transforms. Analytical and numerical results prove that SDCT outperforms both DCT and state-of-the-art directional transforms.
\end{abstract}
\maketitle
\section{Introduction}
\label{sec:intro}

In image and video compression, the two-dimensional discrete cosine transform (2D-DCT) is very popular for its well-known energy compaction properties \cite{ahm74,sayood2012introduction}. The 2D-DCT is obtained applying two separable 1D-DCT transforms along the vertical and horizontal directions, respectively. For this reason, it is very efficient at compressing images in which horizontal or vertical edges are dominating  \cite{zeng2008directional}. Nevertheless, when blocks contain significant directional features and arbitrarily shaped discontinuities the 2D-DCT compression is less efficient \cite{kim2012graph}.

To overcome this problem, various approaches and solutions have been developed \cite{xu2010overview}, most of which consist in modifications of the 2D-DCT in order to incorporate directional information \cite{zeng2008directional,chang2008direction,kamisli2009transforms,cohen2010direction,dre10}.
The Directional DCT (DDCT) presented in \cite{zeng2008directional} is the first attempt in this sense. It consists in a separable transform in which the first 1D-DCT may follow a direction other than the vertical or horizontal one; then the coefficients produced by all directional transforms in the first step are rearranged so that the second transform can be applied to those coefficients that are best aligned with each other.
Later, other works have followed this approach. In  \cite{chang2008direction}, the authors have introduced new directions for the first transform and have proposed a new zigzag scanning method. In \cite{kamisli2009transforms}, it is suggested  to not apply the second-stage DCT, or to apply it only on the DC coefficients generated during the first transform \cite{cohen2010direction}. In \cite{dre10}, DDCT \cite{zeng2008directional} is improved using anisotropic local basis supports, where the optimal basis is selected exploiting the bintree structure of the  dictionary.

These methods, however, have several issues. In particular, they require 1D-DCTs of various lengths, some of which are very short and are not always a power of 2; moreover, the second DCT may not always be applied to coefficients of similar AC frequencies \cite{yeo2012mode}. In our tests, we have also noticed that the performance of the DDCT decreases when the block size increases.

Another method to introduce directionality in the DCT has been presented in \cite{xu2007lifting}, where  directional primary operations have been introduced for the lifting-based DCT. In this way, the DCT-like lifting transform can be applied along any direction, but it extends across block boundaries in order to apply direction adaptation.

In the specific case of intra-frame video coding, another approach has been investigated: the transform is constructed by a directional prediction and a corresponding data-dependent transform. In \cite{ye2008improved}, mode-dependent directional transforms have been derived from Karhunen-Lo\`{e}ve transform, using prediction residuals from training video data. Various follow-up works have then enhanced \cite{ye2008improved} exploiting the symmetry to reduce the number of transform matrices needed \cite{budagavi2010orthogonal, yang2010simplified, tanizawa2010improvement}. \changed{To further improve the performance, several other mode-dependent directional transforms have been proposed, such as the mode-dependent sparse transform \cite{sezer2011robust} and the rate-distortion optimized transform \cite{zhao2012video}. Another data-dependent directional transform called Sparse Orthonormal Transform has been proposed in  \cite{sezer2008sparse} and \cite{sezer2015approximation}. In this case, the image blocks are classified using the image gradient. Then, the transform of each class is optimized by minimizing on a training set an approximation cost. }{In [18], another data-dependent directional transform has been proposed, based on  classification of the blocks using the image gradient. Then, the transform of each class is optimized by maximizing on a training set the sparsity of the transformed vectors.} A common problem of these methods is that training sets must be processed to obtain transforms that are optimal for a given class, so the transform is always dependent on the training set used.

In the last few years, a new approach to image and video coding has been emerging, i.e., transforms on graphs. An image can be viewed as a graph, where each pixel is a node of the graph and the \changed{edges }{arcs} describe the connectivity relations among pixels, e.g., in terms of similarity \cite{shuman2013emerging}. It is possible to define a transform on this domain, called graph Fourier transform \cite{shuman2013emerging}. Thanks to the graph representation, the corresponding transform is ``aware'' of image discontinuities, which are downplayed so as to minimize generation of high-frequency coefficients and maximize energy compaction. \added{Different connectivity patterns lead to different graph transforms. In image applications, the structure of the graph is usually a 4-connected grid graph, where each pixel is connected to its 4 nearest neighbors. This structure has a strong connection with the DCT, because the graph transform of a uniform 4-connected grid graph may be equal to the DCT.} Block-based methods using graph Fourier transform have been proposed in \cite{shen2010edge,kim2012graph,hu2015multiresolution}, but they reported unsatisfactory results on natural images that are not piece-wise smooth.  
\changed{For the specific case of residual coding, a few methods using a graph-based approach have been recently proposed. A novel graph-based method for intra-frame video coding has been presented in \cite{hu2015intra}, which introduces a new generalized graph Fourier transform optimized for intra-prediction residues. }{Recently, a novel graph-based method for intra-frame video coding has been proposed [23], which introduces a new generalized graph Fourier transform optimized for intra-prediction residues.} \added{Instead, in \cite{chao2016graph} the authors propose a block-based lifting transform on graphs for intra-predicted video coding. Moreover, a graph-based method for inter-predicted video coding has been introduced in \cite{egilmez2015graph}, where the authors design a set of simplified graph templates capturing basic statistical characteristics of inter-predicted residual blocks.}
However, one of the main drawbacks of graph-based compression techniques lies in the cost required to represent and encode the graph, which may outweigh the coding gain provided by the edge adaptive transform. For this reason,  some graph-based compression methods that require a small overhead have recently been developed \cite{ioa,iob};  their performance is competitive compared to the DCT, but the price is a high computational cost to construct the transform matrix.   

In our previous work \cite{ioc}, we have presented a new framework for directional transforms. Starting from the graph transform of a grid graph, we have designed a new transform, called steerable DCT (SDCT), which can be obtained  by rotating the 2D-DCT basis by a single given angle for each image block.

In this paper, we analyse the broader problem of finding the best {\em set of rotations} of the 2D-DCT basis for each image block. In particular, we generalize \cite{ioc}  considering an angle for each frequency, which can potentially provide a more compact representation at the price of more side information to transmit. The tradeoff can be analysed from a RD perspective.  We first cast the problem as the minimization of a RD functional. The minimum provides the optimal number of rotation angles per block as well as the angles' values. The problem is well-posed (the global minimum exists), but it is non-convex, hence finding the global minimum  is tricky. The best feasible strategy that one can conceive in such case is iterative alternated minimization, that allows to get to a local minimum or a saddle. This is the basis of our first proposed algorithm, named steerable DCT through alternated minimization (SDCT-AM). If suitably initialized, SDCT-AM is proved to always outperform DCT in RD terms. We have also investigated other strategies to define and transmit the angles' distribution, in order to reduce the angles' transmission cost, and propose a subdivisions into subbands that can be encoded as a binary tree. This is the key idea for our second proposed algorithm, named SDCT-BT, which significantly decreases the amount of side information. Moreover, while in \cite{ioc} the SDCT performance is assessed only in terms of energy compaction, in this paper we develop a fully fledged image encoder to compare the proposed technique with other competing transforms.

The paper is organized as follows. In Section \ref{sec:graph} we define the proposed transform, starting from  graph transforms. Afterwards, in Section \ref{sec:optimal_rotation} we  state our problem in terms of a RD optimization problem, and define the optimal rotation. In Section \ref{sec:algorithms}, we present the SDCT-AM and SDCT-BT algorithms. Section \ref{sec:experimental} is devoted to experimental tests, in which we compare our method to 2D-DCT and directional methods. Finally, in Section \ref{sec:conc} we draw some conclusions.

\section{Steering the DCT}\label{sec:graph}
\subsection{Preliminaries}
We first review some elements of graph signal processing, specifically the concept of graph Fourier transform and its relation to DCT.

We denote an undirected graph as  $\G=(\V,\E)$, where $\V$ is the set of vertices and $\E\subset \V\times \V$ is the set of \changed{edges }{arcs}. Given two graphs $\G_1=(\V_1,\E_1)$ and $\G_2=(\V_2,\E_2)$, let $\G=\G_1 \times \G_2$ be the product graph of $\G_1$ and $\G_2$. Suppose $v_1$, $v_2 \in \V_1$ and $u_1$, $u_2\in \V_2$. Then $(v_1,u_1)$ and $(v_2,u_2)$ are adjacent in $\G$ if and only if one of the following conditions is satisfied \cite{merris1994laplacian}: a) $v_1=v_2$ and $\{u_1,u_2\}\in \E_2$; b) $\{v_1,v_2\}\in \E_1$ and $u_1=u_2$.

For any graph $\G=(\V,\E)$  with $|\V|=N$, we  define the adjacency matrix $A(\G)\in \R^{N\times N}$, where $A(\G)_{ij}=1$ if there is an edge between node $i$ and $j$, otherwise $A(\G)_{ij}=0$. In this paper, we consider undirected graphs with no self loops, that is, $A$ is symmetric and has null diagonal.

\begin{definition}[see \cite{chung1997spectral}] The Laplacian matrix of a graph $\G$ is defined as $L(\G)=\Gamma(\G)-A(\G)\in\R^{N\times N}$, where $A(\G)$ is the adjacency matrix, and $\Gamma(\G)$ is a diagonal matrix with $\Gamma_{ii}$ equal to the number of edges incident to node $i$.
\end{definition}

Any signal $\mathbf{f}\in\R^N$ can be associated with a graph $\G_{\mathbf{f}}=(\V,\E)$ with $|\V|=N$ \cite{shuman2013emerging}; each component $\mathbf{f}_i$, $i=1,\dots,N$ is associated with vertex $v_i\in \V$. On $\G_{\mathbf{f}}$, we define the so-called  graph Fourier transform \cite{shuman2013emerging} of $\mathbf{f}$ as follows: 
\[
\mathbf{c}=U^T \mathbf{f}, 
\]
where $U$ is the matrix whose columns are the eigenvectors of $L(\G_{\mathbf{f}})$. $\mathbf{f}$ can be easily retrieved from $\mathbf{c}$ by inversion: $\mathbf{f}=U \mathbf{c}$.

\begin{figure}[t]
\begin{minipage}[b]{1.0\linewidth}
  \centering
\includegraphics[width=4cm]{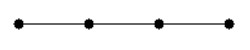}
  \vspace{0.4cm}
  \centerline{(a)}\medskip
\end{minipage}
\begin{minipage}[b]{1\linewidth}
  \centering
\includegraphics[width=4cm]{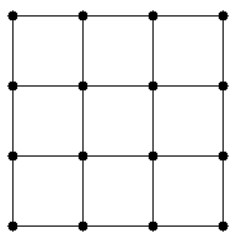}
  \centerline{(b)}\medskip
\end{minipage}
\caption{Two graph models: (a) the path graph $\mathcal P_4$, (b) the square grid graph $\mathcal P_4\times \mathcal P_4$.}
\label{fig:struct}
\end{figure}

One can also recast some  existing transforms as graph Fourier transforms on a specific topology. 
An example is the equivalence between the 1D-DCT and the graph Fourier transform of a path graph. We define a path graph $\mathcal P_N$ as a graph with $N$ vertices and line topology, as shown in Figure \ref{fig:struct}(a). It is known that the eigenvectors of $L(\mathcal P_N)$  are equal to the basis vectors of the 1D-DCT (more precisely DCT-2) \cite{strang1999discrete}. Specifically, the 1D-DCT has $N$ basis vectors $\{\mathbf{v}^{(k)}\}_{k=0}^{N-1}$ which are defined as

\begin{equation}\label{eq1}
\mathbf{v}_{j}^{(k)}=\cos\left (\frac{\pi k}{N} \left (j+\frac{1}{2}\right )\right ), \qquad j,k=0,1,...,N-1.
\end{equation}

Each $\mathbf{v}^{(k)}$ is  the eigenvector of $L(\mathcal P_N)$, for any $k=0,1,...,n-1$, $L$,  associated with the eigenvalue 
\begin{equation}
\lambda_k = 4\sin^2\left ( \frac{\pi k}{2N}\right ).
\label{eq2}
\end{equation}
Given that the multiplicity of the eigenvalues in \eqref{eq2} is always equal to 1, the 1D-DCT basis is the unique eigenbasis for $L(\mathcal P_N)$, therefore the graph Fourier transform for a signal represented by a path graph is equivalent to the 1D-DCT transform.

Let us now consider the product graph of two path graphs, as shown in  Figure \ref{fig:struct}(b). If the two path graphs have the same number of vertices, their product graph $\mathcal P_n\times \mathcal P_n$ is a square grid graph with $N=n^2$ vertices. It has been proved that the basis vectors of the 2D-DCT form an eigenbasis of $L(\mathcal P_n\times \mathcal P_n)$ \cite{zhang2013analyzing}. 

Moreover, the spectrum of the Laplacian of a product graph depends on the spectrum of the two generator graphs, as illustrated in the following theorem.
\begin{theorem}[Theorem 2.21 in \cite{merris1994laplacian}; \cite{merris1998laplacian}]\label{theo:compute_eigen}
 Let $\G_1$ and $\G_2$ be graphs on $N_1$ and $N_2$ vertices, respectively. Then the eigenvalues of $L(\G_1\times \G_2)$ are all possible sums of $\lambda_i(\G_1)+\lambda_j(\G_2)$, with $0\le i\le N_1-1$ and $0\le j\le N_2-1$. Moreover, if $\mathbf{v}^{(i)}$ is an eigenvector of $\G_1$ corresponding to $\lambda_i(\G_1)$, $\mathbf{v}^{(j)}$ an eigenvector of $\G_2$ corresponding to $\lambda_j(\G_2)$, then $\mathbf{v}^{(i)}\otimes \mathbf{v}^{(j)}$ (where $\otimes$ indicates the Kronecker product) is an eigenvector of $\G$ corresponding to $\lambda_i(\G_1)+\lambda_j(\G_2)$.
\end{theorem}
\subsection{Analysis of the eigenvalues' multiplicity}
Leveraging the results presented in the previous paragraph, we build a new transform that can be oriented in any direction. 
Using Theorem \ref{theo:compute_eigen}  and equations \eqref{eq1} and \eqref{eq2}, we can compute the eigenvalues and the  eigenvectors of $L(\mathcal P_n\times \mathcal P_n)$ (which, for simplicity, are labeled with a double index):
\begin{equation}
\lambda _{k,l}=\lambda_k+\lambda_l=4\sin^2\left ( \frac{\pi k}{2n}\right )+4\sin^2\left ( \frac{\pi l}{2n}\right ),
\label{eq3}
\end{equation}
\begin{equation*}
\mathbf{v}^{(k,l)}=\mathbf{v}^{(k)} \otimes \mathbf{v}^{(l)},~~ 0\le k,l \le n-1,
\end{equation*}
where $\mathbf{v}^{(k)}$ is the eigenvector of $\mathcal P_n$ corresponding to $\lambda_k$ and $\mathbf{v}^{(l)}$ is the eigenvector corresponding to $\lambda_l$. From (\ref{eq3}), it is evident that some repeated eigenvalues are present, due to symmetry: $\lambda_{k,l}=\lambda_{l,k}$ for $k\ne l$. Moreover, through straightforward computations, it is possible to prove that the eigenvalue $\lambda=4$ has algebraic multiplicity $n-1$ and corresponds to all eigenvalues $\lambda_{k,n-k}$ with $1\le k\le n-1$. 
Therefore, in the spectrum of $L$ there are only $n-1$ eigenvalues with algebraic multiplicity equal to 1 (i.e. $\lambda_{k,k}$ with $k\ne n/2$), and all the others but $\lambda_{k,n-k}$ have algebraic multiplicity 2. It is important to highlight that even if $\lambda_{k,l}=\lambda_{l,k}$ when $k\ne l$, we still have that $\mathbf{v}^{(k,l)}$ and $\mathbf{v}^{(l,k)}$ are linearly independent, because the Kronecker product is not commutative. Therefore, the geometric multiplicity is equal to the algebraic multiplicity. This means that the dimension of the eigenspaces corresponding to these eigenvalues is bigger than one. This proves the following proposition.
\begin{proposition} The 2D-DCT is not the unique eigenbasis for the Laplacian of a square grid graph.
\end{proposition}

\begin{figure}[t]
\centering
\includegraphics[width=6.4cm]{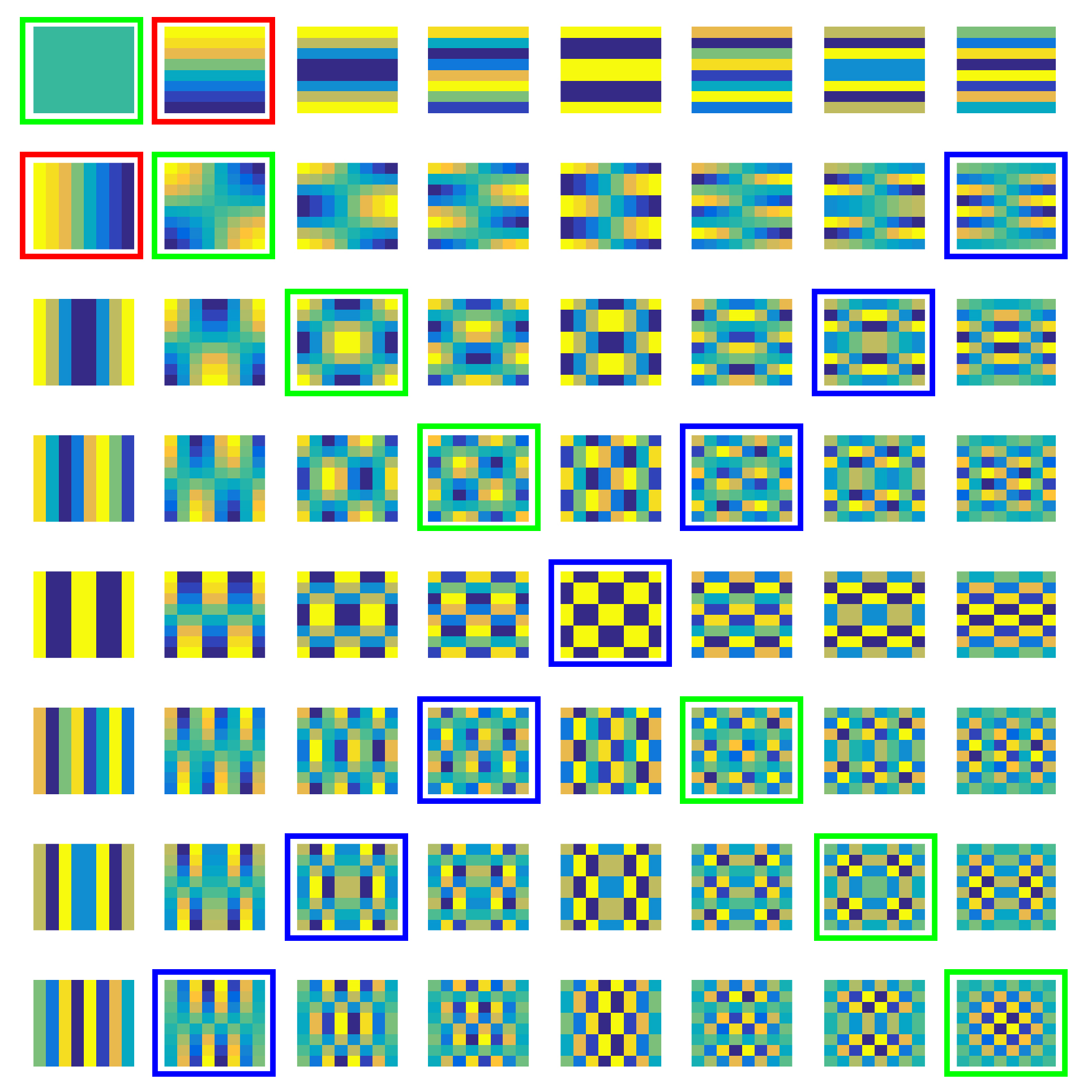}
\caption{2D-DCT basis vectors represented in matrix form (with $n=8$): the corresponding two eigenvectors of an eigenvalue with multiplicity 2 are highlighted in red, the $n-1$ eigenvectors corresponding to $\lambda=4$ are highlighted in blue and the $n-1$ eigenvectors corresponding to the eigenvalues with algebraic multiplicity 1 are highlighted in green.}
\label{fig:basisDCT}
\end{figure}
In Figure \ref{fig:basisDCT} the 2D-DCT basis with $n=8$ is represented in matrix form; as an example, we have highlighted in red the corresponding two eigenvectors of an eigenvalue with multiplicity 2: we can see that they are clearly related to each other, since they represent the same frequency, one in the horizontal direction and the other in the vertical direction.

\subsection{Transform definition}\label{par:transform_definition}
Since the 2D-DCT is not the unique eigenbasis for $L(\mathcal P_n\times \mathcal P_n)$, we aim to find all the other possible eigenbases and choose as transform matrix the one that better fits the properties of the specific image block that we are about to encode.

Given an eigenvalue $\lambda_{k,l}$ of $L(\mathcal P_n\times \mathcal P_n)$ with multiplicity 2 and the two vectors of the 2D-DCT $\mathbf{v}^{(k,l)}$ and $\mathbf{v}^{(l,k)}$ that are the eigenvectors of $L(\mathcal P_n\times \mathcal P_n)$ corresponding to $\lambda_{k,l}$, we can write any other possible basis of the eigenspace corresponding to $\lambda_{k,l}$ as the result of a rotation of $\mathbf{v}^{(k,l)}$ and $\mathbf{v}^{(l,k)}$
\begin{equation}
\begin{bmatrix}
\mathbf{v}^{(k,l)'}\\
\mathbf{v}^{(l,k)'}
\end{bmatrix}
=
\begin{bmatrix}
\cos\theta_{k,l} & \sin\theta_{k,l}\\
-\sin\theta_{k,l} & \cos\theta_{k,l}
\end{bmatrix}
\begin{bmatrix}
\mathbf{v}^{(k,l)}\\
\mathbf{v}^{(l,k)}
\end{bmatrix},
\label{eq:rot}
\end{equation}
where $\theta_{k,l}$ is an angle in $[0,2\pi]$. \added{The rotation described in \eqref{eq:rot} can also be defined as a Givens rotation \cite{golub1996matrix} in the plane described by $\mathbf{v}^{(k,l)}$ and $\mathbf{v}^{(l,k)}$ of the $n^2$-dimensional space.}

For every $\lambda_{k,l}$ with multiplicity 2, we can rotate the corresponding eigenvectors as shown in (\ref{eq:rot}); the $n-1$ eigenvectors corresponding to $\lambda=4$  are rotated in pairs $\mathbf{v}^{(k,n-k)}$ and $\mathbf{v}^{(n-k,k)}$, if $n$ is even $\mathbf{v}^{(\frac{n}{2},\frac{n}{2})}$ is not rotated. In the 2D-DCT matrix, the pairs $\mathbf{v}^{(k,l)}$ and $\mathbf{v}^{(l,k)}$ are replaced with the rotated ones $\mathbf{v}^{(k,l)'}$ and $\mathbf{v}^{(l,k)'}$ obtaining a new transform matrix $V(\theta)\in\R^{n^2\times n^2}$ that can be defined only by the rotation angles used, which we have to transmit to the decoder. The number of angles used is equal to the number of rotated pairs, that is $p=\frac{n(n-1)}{2}$. The new transform matrix $V(\theta)$ can be written as
\[
V(\theta)= V\ROT(\theta),
\]
where $V=V(0)\in\R^{n^2\times n^2}$ is the 2D-DCT transform matrix, $\theta\in\R^{p}$ is the vector containing all the angles used and $\ROT(\theta)\in \R^{n^2\times n^2}$ is the rotation matrix, whose structure is defined so that, for each pair of vectors, it performs the rotation as defined in \eqref{eq:rot}.

$\ROT(\theta)$ can be decomposed in two matrices as 
\[
\ROT(\theta)=\Delta+\widetilde{\ROT}(\theta),
\]
where $\Delta\in\R^{n^2\times n^2}$ is a constant matrix representing the vectors that do not rotate, and $\widetilde{\ROT}(\theta)\in\R^{n^2\times n^2}$ represents the vectors that are rotated. $\Delta$ is a diagonal matrix, with $\Delta_{ii}=1$ for any $i=k n+k$ with $0\le k\le n-1$; otherwise, $\Delta_{ii}=0$. Given $0\le k,l\le n-1$ and $k\ne l$, if $i=kn+l$ and $j=ln+k$, then $\widetilde{\ROT}(\theta)_{ii}=\widetilde{\ROT}(\theta)_{jj}=\cos\theta_{k,l}$, $\widetilde{\ROT}(\theta)_{ij}=\sin\theta_{k,l}$ and $\widetilde{\ROT}(\theta)_{ji}=-\sin\theta_{k,l}$, otherwise $\widetilde{\ROT}(\theta)_{ij}=0$. Then, for any signal $\mathbf{f}\in\R^{n^2}$ our new transform, which in the following will be referred to as SDCT, is defined as follows:
\begin{equation}\label{eq:our_tranform}
\mathbf{c}=V(\theta)^T \mathbf{f}=\ROT(\theta)^TV^T\mathbf{f}=(\Delta^T+\widetilde{\ROT}(\theta)^T)V^T\mathbf{f}.
\end{equation}
\added{Equation \eqref{eq:our_tranform} shows that the SDCT can be decomposed as a product of a rotation matrix $\ROT(\theta)$ and the 2D-DCT transform matrix $V$. Moreover, let $\mathbf{c}_{DCT}\in\R^{n^2}$ be the DCT coefficients of the signal $\mathbf{f}$, then the SDCT can be computed in the following way
\begin{equation}
\label{eq:sep}
\mathbf{c}=\ROT(\theta)^T\mathbf{c}_{DCT}.
\end{equation}
In this way, the complexity of the SDCT can be drastically reduced because $\mathbf{c}_{DCT}$ can be computed using the separability property. Then, to compute the SDCT coefficients, $\mathbf{c}_{DCT}$ is multiplied by the sparse matrix $\ROT(\theta)$.}

\begin{figure}[t]
  \centering
\includegraphics[width=4.5cm]{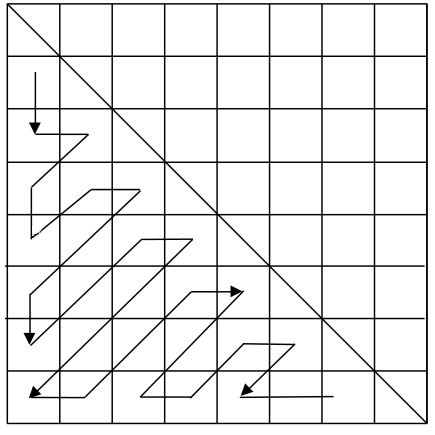}
\caption{Zigzag ordering for the $p$ components of $\theta$.}
\label{fig:ord}
\end{figure}

The components $\theta_{k,l}$ of $\theta$ are ordered using the zigzag pattern shown in Figure \ref{fig:ord}. Unlike the classical zigzag ordering, in this case we consider only $p$ elements, since $\theta_{k,l}=\theta_{l,k}$ and the diagonal elements $\theta_{k,k}$ are not considered, since the eigenvectors $\mathbf{v}^{(k,k)}$ do not rotate.

The transform \eqref{eq:our_tranform} is still the graph transform of a square grid graph, but with a different set of orientations with respect to DCT. As an example, in Figure \ref{fig:basisRDCT}, we show the basis vectors obtained rotating by $\frac{\pi}{4}$ every pair of eigenvectors. As can be seen, the diagonal elements $\mathbf{v}^{(k,k)}$ are the same as the DCT ones because the corresponding eigenvalues have multiplicity one, instead all the others are rotated by $\frac{\pi}{4}$. 
\begin{figure}[t]
\centering
\includegraphics[width=6.6cm]{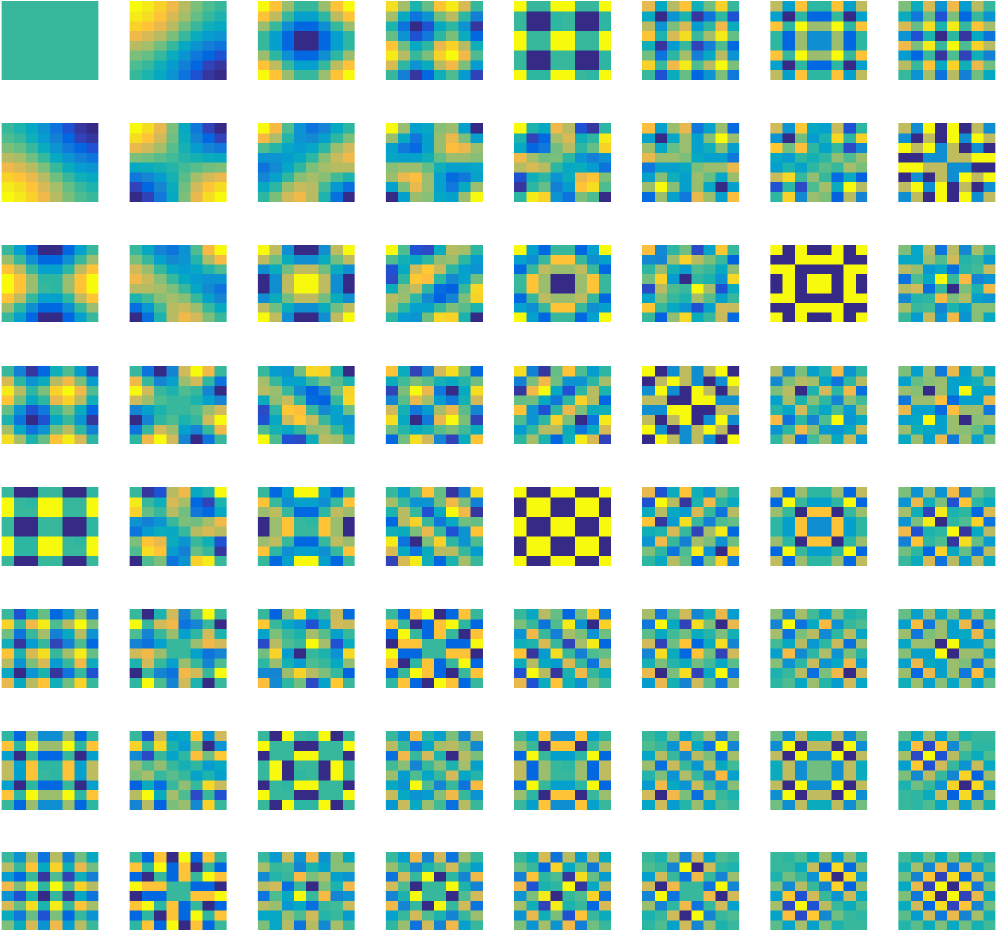}
\caption{Steerable DCT with $\theta=\frac{\pi}{4}$.}
\label{fig:basisRDCT}
\end{figure}

\section{Optimal rotation}\label{sec:optimal_rotation}
In the previous section, we have shown that a new transform can be derived rotating  $2p=n(n-1)$ columns of an $n^2\times n^2$ DCT matrix. The aim of this section is to determine  the set of \emph{optimal} rotation angles under a suitable criterion.

Since our ultimate goal is efficient compression,  a sparse (i.e., with many null coefficients) vector of transform coefficients is desirable. We now illustrate that we can find analytically the rotation $\theta$ that provides the sparsest coefficient representation \eqref{eq:our_tranform}. Let $\mathcal{I}\in\R^{n^2}$ be the original (vectorized) image block. Given an eigenvalue $\lambda_{k,l}$ of $L$ with geometric multiplicity $2$ and its corresponding eigenvectors $\mathbf{v}^{(k,l)}$ and $\mathbf{v}^{(l,k)}$,  for a given block, the corresponding DCT coefficients can be expressed as
$$c_{k,l}=\mathbf{v}^{(k,l)^T}\mathcal{I},$$
$$c_{l,k}=\mathbf{v}^{(l,k)^T}\mathcal{I}.$$
If we rotate this pair of eigenvectors by an angle 
\begin{equation}\label{eq:opt_rot}
\theta_{k,l}=\arctan\frac{c_{k,l}}{c_{l,k}}
\end{equation}
either of the two corresponding coefficients becomes null.
In fact, given $\mathbf{v}^{(k,l)'}$ and $\mathbf{v}^{(l,k)'}$, which are obtained rotating $\mathbf{v}^{(k,l)}$ and $\mathbf{v}^{(l,k)}$ by $\theta_{k,l}$ as in \eqref{eq:rot}, the new transform coefficients are
$$c_{k,l}'=\mathbf{v}^{(k,l)'^T}\mathcal{I},$$
$$c_{l,k}'=\mathbf{v}^{(l,k)'^T}\mathcal{I}.$$

From \eqref{eq:opt_rot}, we conclude that $c_{k,l}'=0$ and all the energy of this coefficient pair is conveyed to $c_{l,k}'$, as shown in Figure \ref{fig:opt}. 
\begin{figure}[t]
  \centering
\includegraphics[width=6cm]{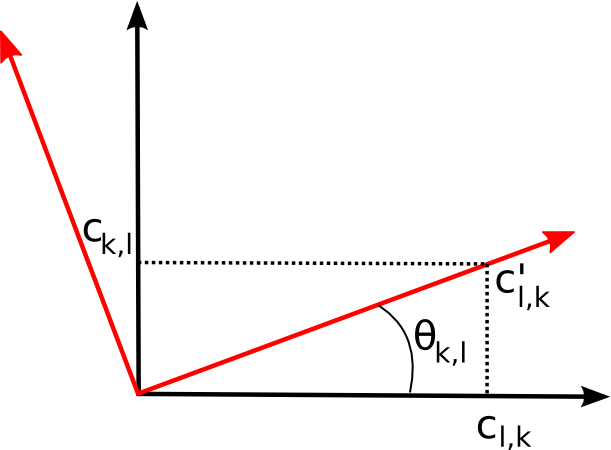}
\caption{Sparsifying rotation: using the angle defined in \eqref{eq:opt_rot} $p$ transform coefficients are exactly null.}
\label{fig:opt}
\end{figure}
\added{This rotation provides the sparsest representation: it exactly nullifies $p$ coefficients.} This is obviously advantageous as it provides a lossless encoding of the image with $n^2-p$ coefficients instead of $n^2$. Nevertheless, since the decoder should also know the $p$ rotation angles  to recover the image,  the total number of values to transmit turns out to be the same. \removed{If quantization is performed on coefficients, DCT produces a lossy representation of the image in which the smallest coefficients are quantized to zero. With SDCT,  quantization error would not affect the zero coefficients, but a price has to be paid in terms of transmission of the angles.} \added{For this motivation, a less sparse solution (i.e., with a smaller number of zero coefficients) or a non-exactly sparse solution (i.e., with many coefficients close to zero, but non exactly zero) might be preferable in RD terms, if it involves less rotation angles.}

The optimal choice of the number and value of rotation angles can be naturally cast as a RD  problem in the variables $\coeff$ and $\theta$.

\subsection{RD  model}
Let $V(\theta)=V\ROT(\theta)$ be the steered transform matrix,  $\mathcal{I}\in\R^{n^2}$ be the  image block, $\theta=(\theta_1,\dots,\theta_p)$ be the ordered set of angles, and $\coeff=(c_1,\dots,c_{n^2})^T$ be the coefficients in the transform \eqref{eq:our_tranform}. As a distortion metric we employ the reconstruction error:
\begin{equation}
D(\coeff,\theta):=\|\mathcal{I}-V(\theta)\coeff\|_2^2.
\end{equation}
We consider two rate contributions, that is, the transform coefficients rate $R_{\coeff}$ and the rotation angles' rate $R_{\theta}$. The total rate is  $R(\coeff,\theta)=R_\mathbf{c}+R_{\theta}$.

In \cite{kim2001novel, dre10, mal98}, it has been shown that for DCT transforms there is an approximately linear relationship between the coding bitrate $R_{\coeff}$ and the $\ell_0$-norm of $\mathbf{c}$, that is, the number of its non-zero coefficients, i.e.
\begin{equation}\label{eq:alpha}
R_{\coeff}=\alpha\|\coeff\|_0
\end{equation}
where $\alpha$ can be empirically found \cite{kim2001novel}. 

Let us now discuss $R_{\theta}$. In \cite{ioc}, we have considered the simple case of using the same angle for all the eigenspaces, and concluded that this is sufficient to outperform classical 2D-DCT. Our aim is now to study the intermediate cases, seeking the optimal number and values of angles yielding the best balance between recovery accuracy and rate. 

Specifically, we split the angles into subbands of DCT coefficients, choosing a single angle for all coefficients in each subband, so that  the vector $\theta$ is piecewise constant. Let $s$ be the number of subbands: if $s<p$, instead of transmitting $p$ angles, we require only $s$ angles values and $s$ indexes indicating where the subvectors end. Assuming no compression for the angles and a quantization over $q_{\theta}$ values in $[0, 2\pi]$ for each angle,  the transmission amounts to $s\lceil \log_2 q_{\theta}\rceil +s \lceil\log_2 p\rceil$, which clearly increases much slower than than $p\lceil \log_2 q_{\theta}\rceil$. We notice that $s$ can be expressed as a function of $\theta$ as follows:
$$s=\|B\theta\|_0$$
where $B\in\R^{p\times p}$ is the discrete difference operator, given by:
\begin{equation*}
B=\left(\begin{array}{ccccc}
1&0&\cdots&\cdots&0\\
-1&1&0&\cdots&0\\
0&\vdots&\ddots&\ddots&0\\
0&\cdots&\cdots&-1&1\\
\end{array}\right).
\end{equation*}

In conclusion, we define the angles rate as follows:
\begin{equation}
R_{\theta}=\|B\theta\|_0 (\lceil\log_2 q_{\theta}\rceil +\lceil \log_2 p \rceil).
\end{equation}

Finally, we assume that both $\coeff$ and $\theta$ are quantized, and denote as $\Q_{\coeff}\subset \R$ and $\Q_{\theta}\subset [0,\pi]$ the respective sets  of available reconstruction values for each component, so that $\coeff\in \Q_{\coeff}^N$ and $\theta \in \Q_{\theta}^p$. We are now ready to define our RD optimization problem. As in \cite{sullivan1998rate}, we consider the following Lagrangian relaxation:
\begin{equation}\label{eq:RD_lagrangian}
\begin{split}
&\min_{\coeff\in\Q_{\coeff}^N,~\theta\in\Q_{\theta}^p} J(\coeff,\theta)\\
&J(\coeff,\theta)=D(\coeff,\theta)+\lambda (R_{\coeff}+R_{\theta})\\
&~~~~~~~~=\|\mathcal{I}-V(\theta)\coeff\|_2^2+\\&~~~~~~~~+\lambda \left[\alpha\|\coeff\|_0  +(\lceil \log_2 q_{\theta}\rceil+\lceil \log_2 p \rceil) \|B\theta\|_0\right],\\
\end{split}
\end{equation}
where $\lambda>0$ is the Lagrangian parameter.

The problem \eqref{eq:RD_lagrangian}  is similar to sparse signal recovery problems, for which hard thresholding techniques can be used  \cite{blu08}. Briefly, a functional of kind $\|A\mathbf{x}-\mathbf{y}\|_2^2+\lambda\|\mathbf{x}\|_0$ with $\mathbf{x}\in\R^n$ and invertible $A\in\R^{n\times n}$  has global minimum at $\mathcal{H}_{\sqrt{\lambda}}[A^{-1}\mathbf{y}]$, where $\mathcal{H}_{\sqrt{\lambda}}:R^n\to \R^n$ is the hard-thresholding operator that sets to zero all the components smaller than $\sqrt{\lambda}$ in magnitude of its input vector. This can be derived as a simpler subcase of iterative hard thresholding for sparse problems \cite[Equation 2.1-2.2]{blu08}: since our transform matrix is orthogonal, the procedure stops after one iteration.  

Our problem is made more difficult by the non-convexity of the distortion term due to the variable $\theta$. However, we remark that the problem is well posed, because it is lower bounded by 0, and it is proper (if $\coeff$ goes to infinity, $J$ tends to infinity as well). This encourages to search  a solution; to this end, we undertake alternated minimization on separated variables. In particular, we notice that the problem can be analytically solved with respect to the individual variables $\coeff$ and $\theta_1,\dots,\theta_p$.

\section{Proposed algorithms}\label{sec:algorithms}
In this section, we present the proposed algorithms SDCT-AM and SDCT-BT to seek the best set of rotations for SDCT.

In the previous section, we have defined the RD optimization problem \eqref{eq:RD_lagrangian} and observed that a global solution is difficult to find due to global non-convexity.  However, the problem is mathematically tractable in the individual variables $\coeff,\theta_1,\theta_2,\dots,\theta_q$ , as we are going to show, and an alternated minimization  achieves a partial optimum (i.e., a local minimum or a saddle point). This is the basis of SDCT-AM.
\subsection{Alternated minimization: SDCT-AM}\label{par:SDCT-AM}

Assuming $\theta$ fixed, the evaluation of $\min_{\coeff\in\Q_{\coeff}^N}J(\coeff,\theta)$ is straightforward. We have

\begin{equation*}
\begin{split}
\min_{\coeff\in\Q_{\coeff}^N}J(\coeff,\theta)&=\min_{\coeff\in\Q_{\coeff}^N}D(\coeff,\theta)+\lambda\alpha \|\coeff\|_0\\
&=\min_{\coeff\in\Q_{\coeff}^N}\sum_{i=1}^{n^2}\left[c_i-(V^T(\theta)\mathcal{I})_i\right]^2+\lambda\alpha \left\|c_i\right\|_0.
\end{split}
\end{equation*}
Therefore, we can solve a separated problem for each component $\coeff_i$, whose solution is given by

$$\mathcal{H}_{\sqrt{\lambda\alpha}}\left[\mathcal{Q}\left[V^T(\theta)\mathcal{I})_i\right]\right]  $$ 
where, for any $x\in\R$, $\mathcal{Q}\left[x\right]$ and $\mathcal{H}_{\sqrt{\lambda\alpha}}\left[x\right]$ respectively indicate the quantization operator that projects onto $\Q_{\coeff}$ and the hard thresholding operator with threshold $\sqrt{\lambda\alpha}$ defined as $\mathcal{H}_{\sqrt{\lambda\alpha}}\left[x\right]=x$ if $|x|>{\sqrt{\lambda\alpha}}$, and 
$\mathcal{H}_{\sqrt{\lambda\alpha}}\left[x\right]=0$ if $|x|\leq {\sqrt{\lambda\alpha}}$.

We notice that $$\left[\mathcal{Q}\left[V^T(\theta)\mathcal{I})_i\right]\right]=\arg\min_{c_i\in\Q_{\coeff}}\left[c_i-(V^T(\theta)\mathcal{I})_i\right]^2$$ since $\left[c_i-(V^T(\theta)\mathcal{I})_i\right]^2$ is convex and symmetric.

The procedure to  minimize $J(\coeff,\theta)$ with respect to $\theta_j$, $j\in\{1,\dots,p\}$, is similar. We have 
\begin{equation*}
\begin{split}
\min_{\theta_j\in\Q_{\theta}}&J(\coeff,\theta)=\hspace{-0.1cm}\min_{\theta_j\in\Q_{\theta}}\hspace{-0.1cm}D(\coeff,\theta)+\lambda(\lceil \log_2 q_{\theta}\rceil+\lceil \log_2 p \rceil)\left\|B\theta\right\|_0\\
\end{split}
\end{equation*}
where the term $\left\|B\theta\right\|_0$ can be substituted by $\left\|\theta_j-\theta_{j+1}\right\|_0+\left\|\theta_j-\theta_{j-1}\right\|_0$ for $j\in\{2,\dots,p-1\}$, by $\left\|\theta_1\right\|_0+\left\|\theta_1-\theta_{2}\right\|_0$ for $j=1$, and by $\left\|\theta_q-\theta_{q-1}\right\|_0$ for $j=p$. 

First, we analytically evaluate $\min_{\theta_j\in[0,2\pi]}D(\coeff,\theta)$. Since $V(\theta)$ is orthogonal for any $\theta$,

\begin{equation*}
\begin{split}
 \|\mathcal{I}-V(\theta)\coeff\|_2^2&= \|\mathcal{I}\|_2^2-2\mathcal{I}^T V(\theta)\coeff+\|\coeff\|_2^2.
\end{split}
 \end{equation*}
Furthermore, it is straightforward to check that  we can define a  matrix $W=W(\coeff)\in\R^{n^2\times 2p}$  such that 
\begin{equation*}
V \widetilde{\ROT}(\theta)\coeff = W(\coeff)(\cos(\theta_1), \sin(\theta_1),\dots, \cos(\theta_p), \sin(\theta_p))^T,
\end{equation*}
$\widetilde{\ROT}(\theta) $ being defined in Section \ref{par:transform_definition}. In this way,
\begin{equation*}
\begin{split}
 V(\theta)\coeff&=V[\Delta +\widetilde{\ROT}(\theta)]\coeff\\
  &=V\Delta\coeff +W(\coeff)(\cos(\theta_1), \sin(\theta_1),\dots)^T.\\
  \end{split}
 \end{equation*}
 Therefore,
\begin{equation*}
 \min_{\theta_j\in[0,\pi]}D(\coeff,\theta)= \min_{\theta_j\in[0,\pi]} -2\mathcal{I}^T W(\coeff)(\cos(\theta_1), \sin(\theta_1),\dots)^T.
\end{equation*}
We then compute the derivative with respect to $\theta_j$, which is equal to zero when $\mathcal{I}^T( W^{(2j)}\sin(\theta_j) -W^{(2j+1)}\cos(\theta_j))=0$, i.e.,

$$\theta_j = \arctan\left(\frac{ \mathcal{I}^T W^{(2j+1)}}{\mathcal{I}^T W^{(2j)}} \right)$$
where $W^{(i)}$ indicates the $i$th column  of $W$. This equation has one solution in $[0,\pi]$, which could be either the maximum or the minimum. For continuity, it suffices to compare this solution with the extreme values $\theta_j=0$ and $\theta_j=\pi$ to obtain the minimum.

Afterwards, as for $\coeff_i$, we proceed by projecting onto $\Q_{\theta_j}$ (again, convexity and symmetry of the subproblem guarantee that $\widehat{\theta}_j=\arg\min_{\theta_j\in\Q_{\theta_j}}D(\coeff,\theta)=\mathcal{Q}[\arg\min_{\theta_j\in[0,\pi]}D(\coeff,\theta)]$). Finally, we perform hard thresholding, which consists in evaluating which one among $\widehat{\theta}_j, \theta_{j-1}, \theta_{j+1}$ is the most convenient choice for $\theta_j$, $j=2,\dots,p-1$ that is, which value provides the minimum $J$. For $j=1$ and $j=p$, clearly the choice is among $\widehat{\theta}_1, 0, \theta_{2}$, and $\widehat{\theta}_p, \theta_{p-1}$.

Alternating these minimization tasks we obtain SDCT-AM, which is summarized in Algorithm 1.

\begin{algorithm}\label{alg:SDCT-AM}
\caption{SDCT-AM}
\label{alg1}
\begin{algorithmic}[1]
\STATE Initialize: $\theta(0), \coeff(0)$;
\FOR{t=1,2,\dots}
\STATE $\coeff(t)=\arg\min_{\coeff\in\Q(\coeff)^N}J(\coeff,\theta)$ (see Section \ref{par:SDCT-AM})
\FOR{$j=p,p-1,\dots,1$}
\STATE $\theta_j(t)=\arg\min_{\theta_j\in\Q(\theta)}J(\coeff,\theta)$ (see Section \ref{par:SDCT-AM})
\ENDFOR
\IF {$J(\coeff(t-1),\theta(t-1))=J(\coeff(t),\theta(t))$}
\STATE \textbf{break}
\ENDIF
\ENDFOR
\end{algorithmic}
\end{algorithm}

\begin{theorem}
There is a time $t_0$ in which $J(c(t),\theta(t))$ in SDCT-AM stabilizes at a  partial optimum.
\end{theorem}

\begin{proof}
The alternated minimization of SDCT-AM guarantees that the sequence $J(\coeff(t),\theta(t))$ is not increasing. Since $J$ is lower bounded by 0 and is a proper function (if $\coeff$ goes to infinity, $J$ tends to infinity), it admits a minimum. Therefore $J(\coeff(t),\theta(t))$ is not increasing and compact, which implies that is convergent. Since $\coeff(t)$ and $\theta(t)$ are quantized values, convergence turns out to be a stabilization, that is, from a time step $t_0$, $J(\coeff(t),\theta(t))$ is constant. Finally, it is easy to check that $(\coeff(t_0),\theta(t_0))$ is partial optimum, because the functional increases moving along the coordinate directions.

\end{proof}

\added{A consequence of this theorem is that the SDCT-AM performance is always better than or equal to the DCT performance, in RD terms. In fact, Since SDCT-AM decreases $J$, it is sufficient to initialize SDCT-AM with DCT to be sure to perform better (or at least equivalently, in the case that DCT is a partial optimum of $J$).

Moreover, the theorem suggests also a stop criterion for SDCT-AM: when $J(\coeff(t),\theta(t))=J(\coeff(t-1),\theta(t-1))$, the algorithm can be stopped. 

}
\subsection{Binary tree for angles structure: SDCT-BT}\label{par:SDCT}

SDCT-AM (Algorithm \ref{alg:SDCT-AM}) is proved to achieve a partial optimum of the RD functional $J$, which is the best results that one can expect to achieve, due to the non-convexity of the problem. In the following we propose an alternative algorithm, called SDCT-BT, which reduces the angles side information cost, allowing more freedom in choosing the rotation angles. Based on the construction of a binary tree to describe the angles subband division, SDCT-BT cannot be theoretically analyzed in terms of a minimization problem, but is experimentally proved to perform well.

Before illustrating SDCT-BT, we specify that in this approach $\coeff$ and $\theta$ are no more considered as separated variables, since $\coeff$ in this case is the vector of the quantized transform coefficients obtained by performing the SDCT: each time we modify $\theta$, we automatically set $\coeff=\mathcal{Q}[V(\theta)^T\mathcal{I}]$, where $\mathcal{Q}$ indicate the operation of quantization onto $\Q_{\coeff}$. Therefore, we will only use the variable $\theta$, and accordingly we will use $J(\theta)$ to indicate the cost functional. 

Moreover, $J(\theta)$ is slightly different from $J(\coeff, \theta)$ in the rate definition. For $R_{\coeff}$, we  use the real bitrate, while  $R_{\theta}$  is determined by the angle selection procedure that we  illustrate in the following.

The angles setting of SDCT-BT is as follows. We start from a single angle value, say one subband, and we iteratively decide if it is convenient to split into different subbands. Specifically, we impose that each subband can be divided into two subbands of equal length if this decreases $J$ (spare pairs of vectors are included in the last group), as shown in Figure \ref{fig:sub}. The decision about splitting a subband is taken by performing an exhaustive search over all possible  $q_{\theta}$ angles and selecting the one minimizing $J$; if the so-obtained $J(\theta)$ is smaller than the current cost $\widehat{J}$, then the split is accepted, and $\widehat{J}=J(\theta)$. We proceed until no more improvement can be obtained, or when the maximum number of subbands is achieved. 

As depicted in Figure \ref{fig:sub_tran}, this procedure is efficient because it can be encoded as a binary decision tree with the root set at level 1. Each node of this tree represents a possible subband and is set to 1 if it actually is a subband, and 0 otherwise. Nodes labeled with 0 are linked to two new nodes, while nodes labeled with 1 are leafs. We represent the final subband subdivision by signaling the decision tree starting from top level 1.

In this way, if the number of subbands is $s$,  the number of nodes in the decision tree is $2s-1$; then we have to signal only $2s-1$ bits. For SDCT-AM the subband structure is encoded over  $s\lceil \log_2 p \rceil= \|B\theta\|_0\lceil \log_2 p \rceil$, which is larger than $2s-1$ for any $p\geq2$.

\begin{figure}[t]
  \centering
  \includegraphics[width=7cm]{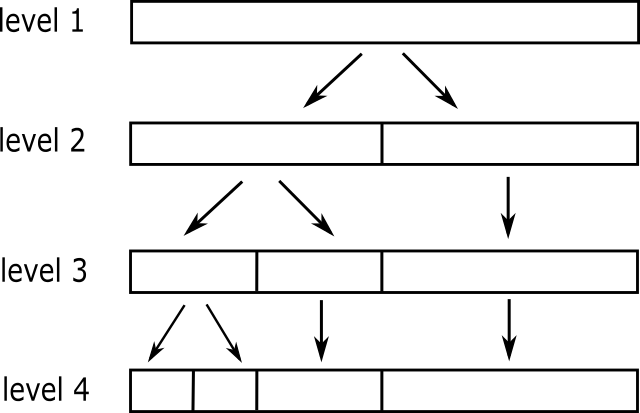}
\caption{Binary subband subdivision for SDCT: from level $1$ downwards, we split a subband if this operation decreases the cost functional $J$}
\label{fig:sub}
\end{figure}
\begin{figure}[t]
  \centering
\includegraphics[width=7.5cm]{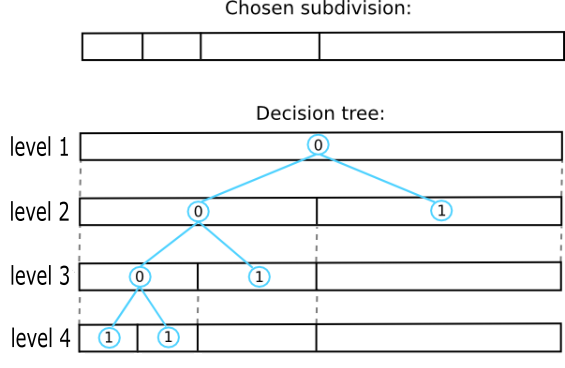}
\caption{Signaling of the subbands structure: from level $1$ downwards, we transmit the labels of the nodes in the binary decision tree.}
\label{fig:sub_tran}
\end{figure}

\begin{algorithm}\caption{SDCT-BT}
\label{alg:SDCT-BT}
\begin{algorithmic}[1]
\STATE Initialize: $k=0$,  $\hat{\theta}=(\theta_0,\theta_0,\dots,\theta_0)$  (i.e. $1$ subband), $\widehat{J}=J(\hat{\theta})$ 
\FOR{$k=1\dots, \lfloor \log_2 p\rfloor$}
\FOR{each subband $s$}
\STATE Split $s$ into two groups of equal length
\STATE $\theta=\widehat{\theta}$
\STATE Sequentially, for each group $g$,  
$$\theta_{j}=\omega~~\text{ for all } j\in g$$
where $\omega=\arg\min_{x\in\Q_{\theta}}J  $ (found via exhaustive search)
\IF{ $J(\theta)<\widehat{J}$}
\STATE $\widehat{J}=J(\theta)$
\STATE $\widehat{\theta}=\theta$
\STATE the two groups are accepted as new subbands
\ENDIF
\ENDFOR
\IF{no split is performed at the current level $k$}
\STATE \textbf{break}
\ENDIF
\ENDFOR
\end{algorithmic}
\end{algorithm}

SDCT-BT is summarized in Algorithm \ref{alg:SDCT-BT}. As one can deduce from Figure \ref{fig:sub_tran}, for each accepted split we use 2 additional bits to signal it.

\subsection{Image codec based on steerable DCT}\label{sub:encoder}
When using SDCT-AM and SDCT-BT, we need to encode three different types of information: the transform coefficients, the rotation angles, and the subband subdivision.
To code the transform coefficients, we perform an uniform quantization and then we code the quantized coefficients using an adaptive bit plane arithmetic coding. 

To code the rotation angles, we fix $q_{\theta}=8$ quantization levels for the angles, uniformly set in $[0,\pi]$ for both SDCT-AM and  SDCT-BT. Then, we use $\log_2 q_{\theta}=3$ bit to transmit each rotation angle. We do not perform any compression on the angles, as their distribution, as observed in our tests, does not exhibit an evident compressibility. \added{In order to improve the compression performance, as future work we may consider a non-uniform angle quantization.}

Regarding the subband subdivision, the two proposed algorithms present two different encoding methods, as explained in the previous part of the section. SDCT-AM requires $(\lceil \log_2 p \rceil)s$ bits, where $p=\frac{n(n-1)}{2}$, and $s$ is the number of subbands.

As also done in \cite{zeng2008directional}, we take into account 1 more bit for each block to declare whether we are applying the directional method or the classical DCT.

\section{Experimental results}\label{sec:experimental}
In this section, we evaluate the performance of the proposed SDCT-AM and SDCT-BT methods and compare them to the state-of-the-art directional transforms. We perform an objective comparison computing the PSNR and a subjective comparison evaluating the SSIM index \cite{wang2004image}. \added{At the end of the section, we also propose some considerations and experiments about a possible future implementation of the SDCT in the HEVC standard.}

We test  SDCT-AM (Algorithm \ref{alg:SDCT-AM}) and SDCT-BT (Algorithm \ref{alg:SDCT-BT}) on some standard grayscale images and on intra-frame prediction errors. For the prediction errors, we use HEVC to generate intra-frame prediction residuals on the first frame of few test video sequences. For both images and residual frames, we use different block sizes $n\times n$ with $n\in\{8,16,32\}$. 
We compare their performance against the classical DCT, the Directional DCT \cite{zeng2008directional} and the SDCT with only one rotation angle per block (SDCT-1), as proposed in \cite{ioc}. \added{ In Section \ref{sec:wave} we also show a brief comparison between wavelets and SDCT.}

For the DDCT and the SDCT-1, we code the transform coefficients using the same method used for SDCT-AM and SDCT-BT (see Sec. \ref{sub:encoder}); in addition to the bitrate of the coefficients, we count 3 bit per block to transmit the chosen angle and one additional bit to signal if we are using the directional method or the classical DCT. \added{Regarding the wavelets, we use CDF 9/7 wavelets and we code the transform coefficients with the same method used for the other transforms.}

For all our simulations, we consider $q_{\theta}=8$ angles uniformly set in $[0,\pi]$, as explained in Section \ref{sub:encoder}. We initialize both SDCT-AM and SDCT-BT with one single angle, testing all 8 possible initializations and eventually choosing the best one. For SDCT-BT, the maximum number of iterations is set by $\lfloor \log_2 p\rfloor$, while for SDCT-AM we get a stationary point in very few iterations (less than 10).

For SDCT-AM, we need to select the parameter $\alpha$ defined in \eqref{eq:alpha}. As we do not know $R_{\coeff}$ and $\|\coeff\|_0$ in advance, we employ the values of $R_{\coeff}$ and $\|\coeff\|_0$ estimated by the classical DCT, multiplied by 2 (we observe in fact that slight overestimation is more safe).

\subsection{Objective comparison}
\begin{table}[t]
\centering

\caption{Average gain in PSNR with respect to DCT measured with Bjontegaard metric (tests on images)}
\label{tab:bjo_classic}
\resizebox{\columnwidth}{!}{  \begin{tabular}{cc|cccc}
Image & block size & DDCT & SDCT-1 & SDCT-AM & SDCT-BT\\
\hline\hline
 House& 8$\times$8 & 0.325 & 0.382 & 0.406 &0.432\\
 256$\times$256& 16$\times$16 & 0.274 & 0.335 &  0.636 &0.563\\
 & 32$\times$32 & 0.312 & 0.259 &0.718 &0.603\\
\hline
Barbara& 8$\times$8 & 0.285 & 0.288 & 0.328 &0.321\\
512$\times$512 & 16$\times$16 & 0.153 & 0.195 &0.507 &0.392\\
 & 32$\times$32 & 0.074 & 0.093 & 0.567 &0.448\\
\hline
Boat & 8$\times$8 & 0.238 & 0.271 & 0.330 &0.301\\
512$\times$512 & 16$\times$16 & 0.105 & 0.160 &  0.499 &0.338\\
 & 32$\times$32 & 0.043 & 0.076 & 0.565 &0.392\\
\hline
Lena & 8$\times$8 & 0.349 & 0.347 &  0.375 &0.378\\
512$\times$512 & 16$\times$16 & 0.260 & 0.252 & 0.578 &0.460\\ 
 & 32$\times$32 & 0.170 & 0.129 & 0.624 &0.519\\
\hline
Aerial & 8$\times$8 & 0.343& 0.490 &  0.476 &0.572\\
256$\times$256 & 16$\times$16 & 0.132 & 0.297 & 0.512 &0.720\\ 
 & 32$\times$32 & 0.017 & 0.143 & 0.455 &0.985\\
\hline
Stream  & 8$\times$8 & 0.394 & 0.417 &  0.442 &0.476\\
512$\times$512 & 16$\times$16 & 0.165 & 0.256 & 0.547 &0.559\\ 
 & 32$\times$32 & 0.046 & 0.119 & 0.522 &0.736\\
\hline
Couple & 8$\times$8 & 0.239 & 0.294 &  0.341 &0.326\\
256$\times$256 & 16$\times$16 & 0.140 & 0.223 & 0.570 &0.456\\ 
 & 32$\times$32 & 0.066 & 0.114 & 0.620 &0.630\\
\hline
F16 & 8$\times$8 & 0.286 & 0.417 &  0.404 &0.459\\
512$\times$512 & 16$\times$16 & 0.198 & 0.340 & 0.620 &0.632\\ 
 & 32$\times$32 & 0.094 & 0.181 & 0.631 &0.729\\
\hline
\end{tabular}
}
\end{table}

\begin{table}[t]
\centering
\caption{Average gain  in PSNR with respect to DCT measured with Bjontegaard metric (tests on intra-prediction errors)}
\label{tab:bjo_video}
\resizebox{\columnwidth}{!}{ \begin{tabular}{cc|cccc}
Prediction residual & block size & DDCT & SDCT-1 & SDCT-AM & SDCT-BT\\
\hline\hline
RaceHorses & 8$\times$8 & 0.401 & 0.443 & 0.431  &0.477\\
 416$\times$240 & 16$\times$16 & 0.249 & 0.313 &  0.461 &0.625\\
 & 32$\times$32 & 0.119 & 0.164 & 0.354 &0.827\\
\hline
RaceHorses  & 8$\times$8 & 0.407 & 0.431 & 0.455 &0.459\\
 832$\times$480 & 16$\times$16 & 0.228 & 0.278 & 0.527 &0.549\\
& 32$\times$32 & 0.125 & 0.138 & 0.461 &0.776\\
\hline
BasketballPass & 8$\times$8 & 0.322 & 0.381 & 0.503 & 0.415\\
416$\times$240 & 16$\times$16 & 0.200 & 0.235 &0.619 &0.502\\
 & 32$\times$32 & 0.120 & 0.133 & 0.606 &0.652\\
\hline
PartyScene & 8$\times$8 &0.468 &0.368 &0.335 &0.388 \\
 832$\times$480 & 16$\times$16 &0.283 &0.235 &0.307 &0.451\\ 
 & 32$\times$32 &0.138 &0.122 & 0.243 &0.549\\
\hline
ChinaSpeed & 8$\times$8 &0.613 &0.391 &0.431 &0.382 \\
 1024$\times$768 & 16$\times$16 &0.486 &0.312 &0.565 &0.477 \\ 
 & 32$\times$32 &0.289 &0.150 & 0.491&0.527\\
\hline
Keiba & 8$\times$8 &0.207 &0.380 &0.455 & 0.435 \\
 416$\times$240 & 16$\times$16 &0.117 &0.226 & 0.507&0.546 \\ 
 & 32$\times$32 &0.078 &0.098 & 0.470 & 0.770\\
\hline
Keiba & 8$\times$8 &0.267 &0.331 &0.471&0.367 \\
 832$\times$480 & 16 $\times $16 &0.157 &0.205 &0.580 &0.419 \\ 
 & 32$\times$32 &0.068 &0.086 & 0.543 & 0.510\\
\hline
Kristen$\&$Sara & 8$\times$8 &0.265 &0.264 &0.417&0.273 \\
 1280$\times$720 & 16$\times$16 &0.217 &0.220 &0.607 &0.396 \\ 
 & 32$\times$32 &0.124 &0.129 & 0.644&0.529\\
\hline
\end{tabular}}
\end{table}
In Tables \ref{tab:bjo_classic} and \ref{tab:bjo_video},  we summarize our performance results in terms of average gain in PSNR compared to DCT, evaluated through the Bjontegaard metric \cite{bjontegaard2001calcuation}. 

In Table \ref{tab:bjo_classic}, the comparison is performed on eight classical grayscale images (House, Barbara, Boat, Lena, Aerial 5.1.10 \cite{img_db}, Stream and  Bridge 5.2.10 \cite{img_db}, Couple 4.1.02 \cite{img_db}, Airplane F16 4.2.05 \cite{img_db}; color images have been converted to grayscale). The gains obtained by DDCT and SDCT-1 are similar, and decrease as the block size increases. An inverse behavior characterizes SDCT-AM and SDCT-BT, which generally improve using larger blocks. For blocks $8\times8$, the four methods are quite similar, while for large blocks SDCT-AM and SDCT-BT are definitely preferable than DDCT and SDCT-1.
The PSNR gain ranges from 0.3 dB to nearly 1 dB.

In Table \ref{tab:bjo_video}, prediction errors are considered on eight different videos. The behavior is similar to that appreciated for images in  Table \ref{tab:bjo_classic}: the gain obtained by SDCT-AM and SDCT-BT with respect to DDCT and SDCT-1 is more consistent as the block size increases. In this case, the PSNR gain ranges from 0.3 dB to 0.8 dB.

\added{From the results we can see that the performance of SDCT-AM and SDCT-BT are similar. In certain cases (such as Boat or Barbara), SDCT-AM outperforms SDCT-BT.} \changed{Instead, in other cases }{In some cases,} the performance of SDCT-AM slightly decreases using larger block sizes, while that of SDCT-BT always increases. This happens mostly with prediction errors and with textured images (such as Aerial), for which the non-regularity may require a higher number of subbands. In such frameworks, SDCT-AM is penalized as it uses a larger number of bits to signal the subbands structure if compared to SDCT-BT.

\begin{figure}[h]
\begin{minipage}[b]{\linewidth}
  \centering
\includegraphics[width=7.5cm]{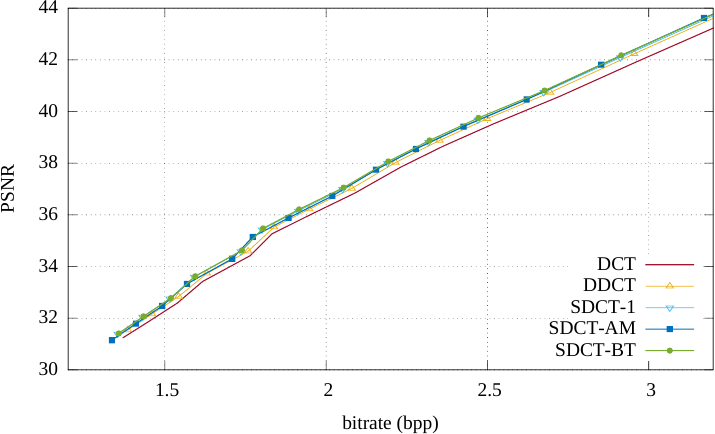}
\end{minipage}
\begin{minipage}[b]{\linewidth}
  \centering
\includegraphics[width=7.5cm]{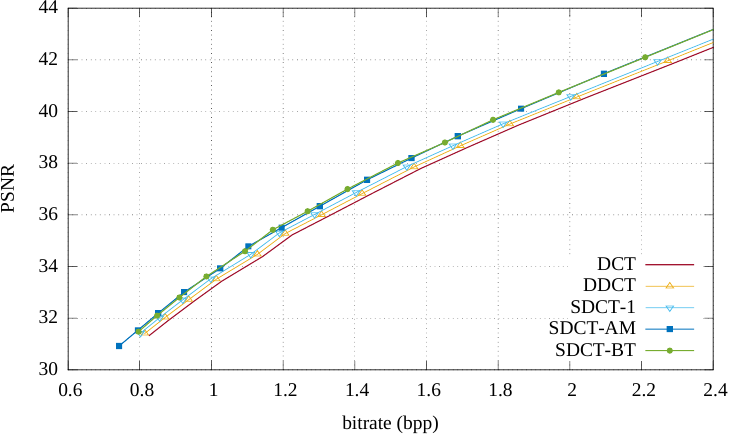}
\end{minipage}
\begin{minipage}[b]{\linewidth}
  \centering
\includegraphics[width=7.5cm]{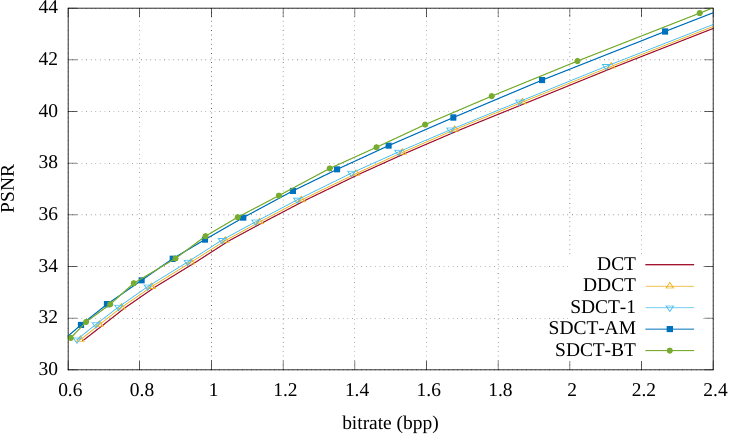}
\end{minipage}
\caption{RD performance comparison for the image Airplane F16 using different block sizes: from top to bottom, $n=8,16,32$}
\label{fig:rd_barb}
\end{figure}

In Figure \ref{fig:rd_barb}, we depict the RD curves concerning the image Airplane F16, for $n=8,16,32$. For $n=16,32$, SDCT-AM and SDCT-BT turn out to be better than the state-of-the-art methods. 

\added{
\subsection{Wavelet comparison}
\label{sec:wave}
For still image compression, coding schemes based on wavelets have achieved significantly better performance compared to DCT-based compression methods \cite{taubman2012jpeg2000}. Instead in video coding, wavelet-based compression methods have not shown significant performance gains versus DCT-based methods \cite{sikora2005trends}. In our work we consider both images and videos, but our focus is mainly on video compression and a possible future implementation of the SDCT in a video compression standard. For this reason, we use as main benchmark the DCT, that is the core transform of most video standards. However, we also present a comparison between wavelets and SDCT on a few sample images.

In Table \ref{tab:wav} we show a comparison between SDCT and wavelets for $n=16,32,64$. We evaluate the performance of SDCT-AM and SDCT-BT in terms of average gain in PSNR compared to wavelets. As we can see from the results, when the dimension of the block is small, SDCT-AM and SDCT-BT show a significant quality gain. Instead, at larger block size the wavelets usually outperform both SDCT-AM and SDCT-BT. It is interesting to point out that in the test Aerial with $n=64$ SDCT-BT outperforms the wavelets, which in turn outperform the classical DCT. This demonstrates that sometimes the improvement obtained by SDCT is significant to make the DCT approach more efficient than other approaches.}
\begin{table}[t]
\centering
\caption{Average gain in PSNR with respect to wavelets measured with Bjontegaard metric}
\label{tab:wav}
\resizebox{\columnwidth}{!}{ 
\begin{tabular}{cc|ccc}
Prediction residual & block size  &  SDCT-AM & SDCT-BT & DCT\\
\hline\hline
 Boat& 16$\times$16 &1.858&1.702&1.359\\
 512$\times$512 & 32$\times$32 & 1.589&1.384&0.985\\
  &64$\times$64&-2.183  &-1.862&-3.199\\
\hline
 Aerial& 16$\times$16 &1.724   & 1.519&1.012\\
256$\times$256 & 32$\times$32 & 1.161 &1.690&0.714\\
 & 64$\times$64 & -0.432&0.294& -0.886\\
\hline
 Stream& 16$\times$16 &1.265 &1.272 &0.718\\
 512$\times$512 & 32$\times$32 & 1.043 & 1.234 &0.495\\ 
 & 64$\times$64  &-0.802  &-0.232 &-1.364\\
\hline
\end{tabular}}
\end{table}

\begin{figure}[h!]
\begin{minipage}[b]{\linewidth}
  \centering
\includegraphics[width=7.5cm]{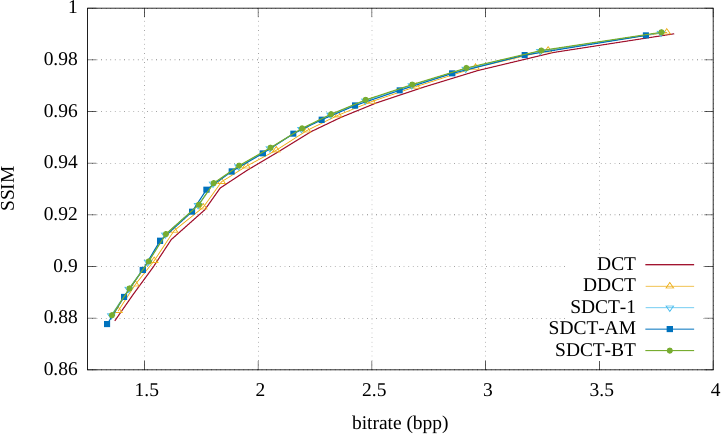}
\end{minipage}
\begin{minipage}[b]{\linewidth}
  \centering
\includegraphics[width=7.5cm]{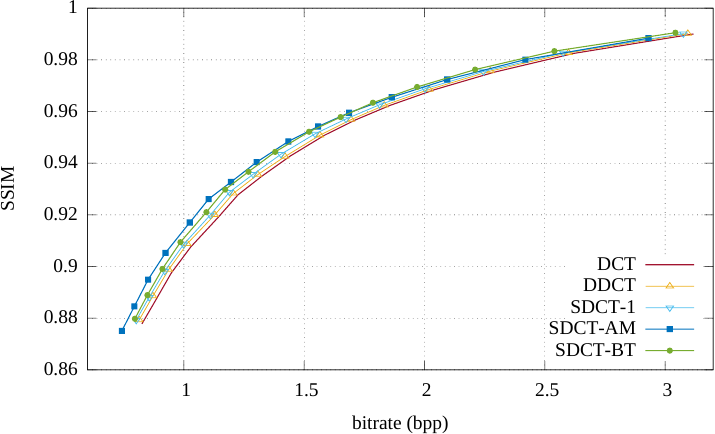}
\end{minipage}
\begin{minipage}[b]{\linewidth}
  \centering
\includegraphics[width=7.5cm]{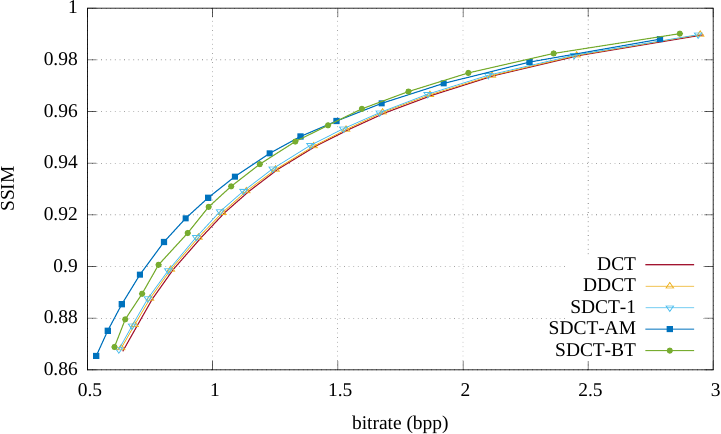}
\end{minipage}
\caption{SSIM performance comparison for the image Airplane F16 using different block sizes: from top to bottom, $n=8,16,32$}
\label{fig:ssim_barb}
\end{figure}

\subsection{Subjective comparison}

Since the PSNR is not always a good representation of the visual quality, we also compute the SSIM index in order to evaluate the perceived quality. The results for the image Barbara are shown in Figure \ref{fig:ssim_barb}. Also in this case, we can see that when we use smaller blocks the performance of the three directional methods are very similar, instead when the block size increases the SDCT clearly outperforms the other methods. 

\begin{figure}[h]
\centering
\begin{minipage}[b]{0.45\linewidth}
  \centering
\includegraphics[width=3.3cm]{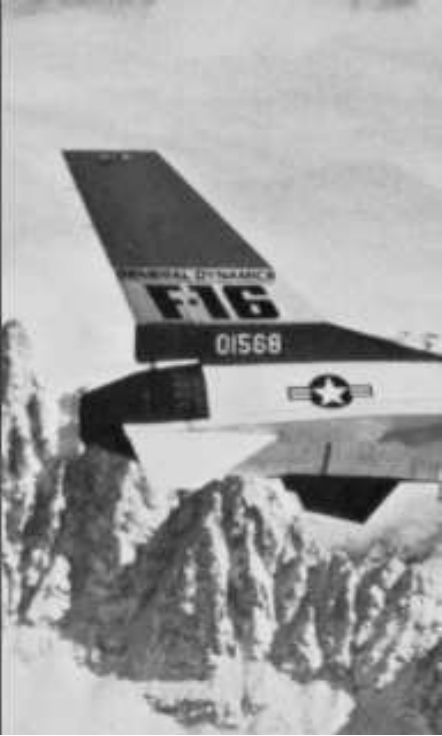}
Original image
\end{minipage}
\hfill
\begin{minipage}[b]{0.45\linewidth}
 \centering
\includegraphics[width=3.3cm]{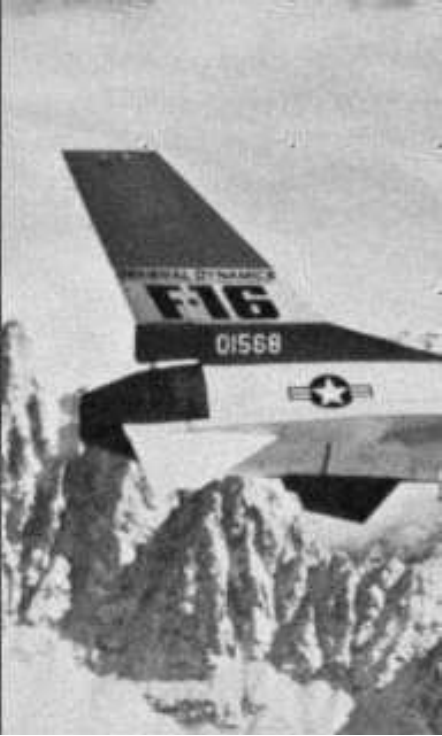}
DCT
\end{minipage}\\
\begin{minipage}[b]{0.45\linewidth}
  \centering
\includegraphics[width=3.3cm]{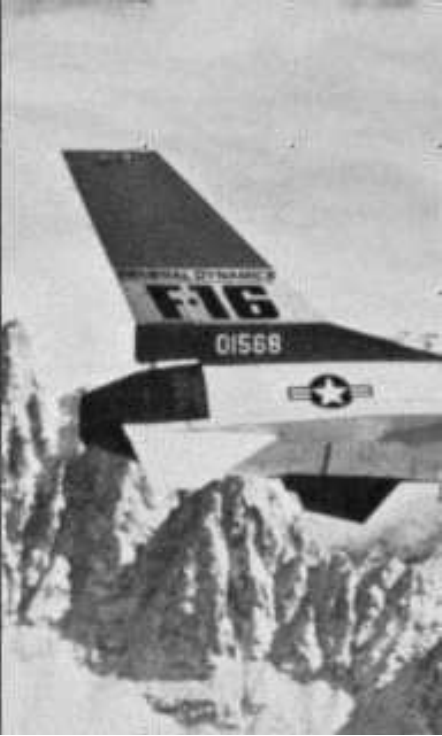}
SDCT-AM
\end{minipage}
\hfill
\begin{minipage}[b]{0.45\linewidth}
  \centering
\includegraphics[width=3.3cm]{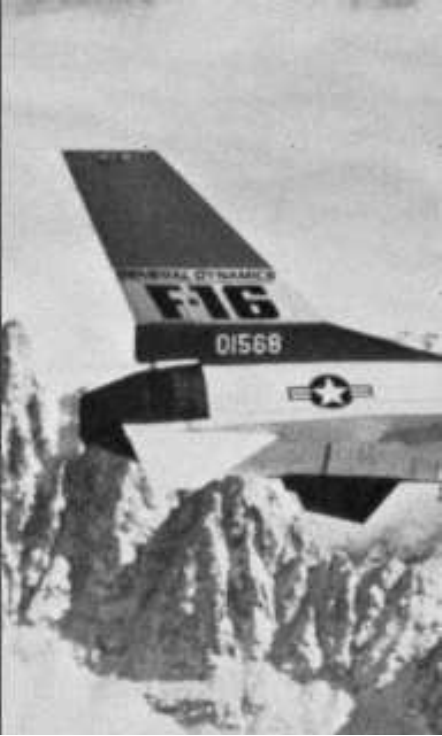}
SDCT-BT
\end{minipage}
\caption{Visual comparison on Airplane F16 image (block size 64$\times$64, 2 bpp): at the same bpp, DCT is more spotted than  SDCT-AM and SDCT-BT.}
\label{fig:visage_barb}
\end{figure}

In Figure \ref{fig:visage_barb}, we show a detail of F16 (block size 64$\times$64,  0.8 bpp) in which a visual improvement can be observed in SDCT-AM and SDCT-BT with respect to DCT. 

\added{\subsection{Future applications}
To conclude the experimental section, we propose some observations and tests regarding possible future applications of SDCT. In particular, we investigate the possibility to implement efficiently the proposed SDCT in the HEVC standard. 

In HEVC, the core transform is DCT \cite{budagavi2013core} \cite{wien2015high}. Replacing it with SDCT is then expected to produce a performance improvement. A test implementation of SDCT within HEVC is beyond the purpose of this paper and is left for future work. However, it is worth mentioning that HEVC uses an integer version of  DCT, i.e. an approximate DCT that can be stored using only integer values \cite{wien2015high}. This clearly has memory advantages, but involves a not exactly orthogonal transform. For this reason, we have tested SDCT using the same integer approximation in order to evaluate the possible drawbacks. To compute the integer approximation of the proposed SDCT we have used equation \eqref{eq:sep}, where the DCT coefficients are computed using the integer DCT defined in HEVC. The obtained results (of which we show just some samples in Table \ref{tab:integer}) are in line with the previous non-integer approach. Moreover, Table \ref{tab:perc} shows that in mostly all cases the proposed SDCT is chosen in a significant number of blocks.

This is a first step that suggests the possibilty to implement efficiently an integer SDCT in the HEVC standard.}

\begin{table}[t]
\centering
\caption{Integer SDCT for HEVC: average gain in PSNR with respect to integer DCT measured with Bjontegaard metric}
\label{tab:integer}
\resizebox{\columnwidth}{!}{ \begin{tabular}{cc|cc}
Prediction residual & block size  & integer SDCT-AM & integer SDCT-BT\\
\hline\hline
RaceHorses & 8$\times$8 &  0.429  &0.476\\
 416$\times$240 & 16$\times$16 & 0.458 &0.613\\
 & 32$\times$32 & 0.353 &0.830\\
\hline
BasketballPass & 8$\times$8 & 0.498 & 0.413\\
416$\times$240 & 16$\times$16 &0.615 &0.494\\
 & 32$\times$32 & 0.609 &0.642\\
\hline
Keiba & 8$\times$8 &0.452 &0.431 \\
 416$\times$240 & 16$\times$16 &0.505 &0.538\\ 
 & 32$\times$32  & 0.472 &0.754\\
\hline
\end{tabular}}
\end{table}

\begin{table}[t]
\centering
\caption{Percentage of blocks where the SDCT is chosen over the DCT at 40 $\mathrm{d}$B}
\label{tab:perc}
\resizebox{\columnwidth}{!}{ 
\begin{tabular}{cc|cc}
Prediction residual & block size  & integer SDCT-AM & integer SDCT-BT\\
\hline\hline
RaceHorses & 8$\times$8 & 50\%  &47\%\\
 416$\times$240 & 16$\times$16 &56\%  &71\%\\
 & 32$\times$32 &  38\%&84\%\\
\hline
BasketballPass & 8$\times$8 & 64\%  & 46\%\\
416$\times$240 & 16$\times$16 &80\% &65\%\\
 & 32$\times$32 & 77\% &78\%\\
\hline
Keiba & 8$\times$8 & 57\% & 44\%\\
 416$\times$240 & 16$\times$16 &79\% &60\%\\ 
 & 32$\times$32  & 89\%  & 80\%\\
\hline
\end{tabular}}
\end{table}

\section{Conclusions}\label{sec:conc}
Exploiting the properties of the graph transform of a grid graph, we have introduced  a new 2D-DCT that can be steered in any chosen direction, selecting different directions for different frequencies. We have introduced a RD optimization problem, whose solution provides the optimal number of rotation angles per block and their values. We have proposed the SDCT-AM algorithm, that achieves a partial optimum, and the SDCT-BT algorithm, that leverages a binary decision tree to lighten the load due to angles signaling. Comparisons with DCT and DDCT show that our methods achieves a significant gain for image blocks of size larger than $16\times16$.

\section*{Acknowledgement}
{This work has been supported by Sisvel Technology, and by the European Research Council under FP7 / ERC, Grant agreement n.279848 - CRISP project.}



\begin{thebibliography}{10}
\providecommand{\url}[1]{#1}
\csname url@samestyle\endcsname
\bibitem{ahm74}
N.~Ahmed, T.~Natarajan, and K.~Rao, ``Discrete cosine transform,'' \emph{IEEE
  Trans. Computers}, vol. C-23, no.~1, pp. 90--93, 1974.

\bibitem{sayood2012introduction}
K.~Sayood, \emph{Introduction to data compression}, 4th~ed.\hskip 1em plus
  0.5em minus 0.4em\relax Elsevier, 2012.

\bibitem{zeng2008directional}
B.~Zeng and J.~Fu, ``Directional discrete cosine transforms - a new framework
  for image coding,'' \emph{IEEE Trans. Circuits Syst. Video Technol.},
  vol.~18, no.~3, pp. 305--313, 2008.

\bibitem{kim2012graph}
W.~S. Kim, S.~K. Narang, and A.~Ortega, ``Graph based transforms for depth
  video coding,'' in \emph{Proc. IEEE International Conference on Acoustics,
  Speech, and Signal Processing (ICASSP)}, 2012, pp. 813--816.

\bibitem{xu2010overview}
J.~Xu, B.~Zeng, and F.~Wu, ``An overview of directional transforms in image
  coding,'' in \emph{Proc. IEEE International Symposium on Circuits and Systems
  (ISCAS)}, 2010, pp. 3036--3039.

\bibitem{chang2008direction}
C.~L. Chang and B.~Girod, ``Direction-adaptive partitioned block transform for
  image coding,'' in \emph{Proc. IEEE International Conference on Image
  Processing (ICIP)}, 2008, pp. 145--148.

\bibitem{kamisli2009transforms}
F.~Kamisli and J.~S. Lim, ``Transforms for the motion compensation residual,''
  in \emph{Proc. IEEE International Conference on Acoustics, Speech, and Signal
  Processing (ICASSP)}, 2009, pp. 789--792.

\bibitem{cohen2010direction}
R.~A. Cohen, S.~Klomp, A.~Vetro, and H.~Sun, ``Direction-adaptive transforms
  for coding prediction residuals,'' in \emph{Image Processing (ICIP), 2010
  17th IEEE International Conference on}, 2010, pp. 185--188.

\bibitem{dre10}
A.~Dr\'emeau, C.~Herzet, C.~Guillemot, and J.~Fuchs, ``Sparse optimization with
  directional dct bases for image compression,'' in \emph{IEEE International
  Conference on Acoustics Speech and Signal Processing (ICASSP)}, 2010, pp.
  1290--1293.

\bibitem{yeo2012mode}
C.~Yeo, Y.~H. Tan, Z.~Li, and S.~Rahardja, ``Mode-dependent transforms for
  coding directional intra prediction residuals,'' \emph{IEEE Trans. Circuits
  and Systems for Video Technology}, vol.~22, no.~4, pp. 545--554, 2012.

\bibitem{xu2007lifting}
H.~Xu, J.~Xu, and F.~Wu, ``Lifting-based directional {DCT}-like transform for
  image coding,'' \emph{IEEE Trans. Circuits Syst. Video Technol.}, vol.~17,
  no.~10, pp. 1325--1335, 2007.

\bibitem{ye2008improved}
Y.~Ye and M.~Karczewicz, ``Improved {H.264} intra coding based on
  bi-directional intra prediction, directional transform, and adaptive
  coefficient scanning,'' in \emph{Proc. IEEE International Conference on Image
  Processing (ICIP)}, 2008, pp. 2116--2119.

\bibitem{budagavi2010orthogonal}
M.~Budagavi and M.~Zhou, ``Orthogonal {MDDT} and mode dependent {DCT},''
  \emph{ITU-T Q}, vol.~6, 2010.

\bibitem{yang2010simplified}
H.~Yang, J.~Zhou, and H.~Yu, ``Simplified {MDDT} ({SMDDT}) for intra prediction
  residual,'' \emph{Doc. JCTVC-B039, MPEG-H/JCT-VC}, 2010.

\bibitem{tanizawa2010improvement}
A.~Tanizawa, J.~Yamaguchi, T.~Shiodera, T.~Chujoh, and T.~Yamakage,
  ``Improvement of intra coding by bidirectional intra prediction and 1
  dimensional directional unified transform,'' \emph{Doc. JCTVC-B042,
  MPEG-H/JCT-VC}, 2010.

\bibitem{sezer2011robust}
O.~G. Sezer, R.~Cohen, and A.~Vetro, ``Robust learning of 2-d separable
  transforms for next-generation video coding,'' in \emph{Proc. IEEE Data
  Compression Conference (DCC)}, 2011, pp. 63--72.

\bibitem{zhao2012video}
X.~Zhao, L.~Zhang, S.~Ma, and W.~Gao, ``Video coding with rate-distortion
  optimized transform,'' \emph{IEEE Trans. Circuits Syst. Video Technol.},
  vol.~22, no.~1, pp. 138--151, 2012.

\bibitem{sezer2008sparse}
O.~G. Sezer, O.~Harmanci, and O.~G. Guleryuz, ``Sparse orthonormal transforms
  for image compression,'' in \emph{Proc. IEEE International Conference on
  Image Processing (ICIP)}, 2008, pp. 149--152.

\bibitem{sezer2015approximation}
O.~G. Sezer, O.~G. Guleryuz, and Y.~Altunbasak, ``Approximation and compression
  with sparse orthonormal transforms,'' \emph{IEEE Trans. Image Process.},
  vol.~24, no.~8, pp. 2328--2343, 2015.

\bibitem{shuman2013emerging}
D.~I. Shuman, S.~K. Narang, P.~Frossard, A.~Ortega, and P.~Vandergheynst, ``The
  emerging field of signal processing on graphs: Extending high-dimensional
  data analysis to networks and other irregular domains,'' \emph{IEEE Signal
  Process. Mag.}, vol.~30, no.~3, pp. 83--98, 2013.

\bibitem{shen2010edge}
G.~Shen, W.~S. Kim, S.~K. Narang, A.~Ortega, J.~Lee, and H.~Wey,
  ``Edge-adaptive transforms for efficient depth map coding,'' in \emph{Picture
  Coding Symposium (PCS)}, 2010, pp. 2808--2811.

\bibitem{hu2015multiresolution}
W.~Hu, G.~Cheung, A.~Ortega, and O.~C. Au, ``Multiresolution graph fourier
  transform for compression of piecewise smooth images,'' \emph{IEEE Trans.
  Image Process.}, vol.~24, no.~1, pp. 419--433, 2015.

\bibitem{hu2015intra}
W.~Hu, G.~Cheung, and A.~Ortega, ``Intra-prediction and generalized graph
  fourier transform for image coding,'' \emph{IEEE Signal Process. Lett.},
  vol.~22, no.~11, 2015.

\bibitem{chao2016graph}
Y.~H. Chao, A.~Ortega, and S.~Yea, ``Graph-based lifting transform for
  intra-predicted video coding,'' in \emph{Proc. IEEE International Conference
  on Acoustics, Speech and Signal Processing (ICASSP)}, 2016, pp. 1140--1144.

\bibitem{egilmez2015graph}
H.~E. Egilmez, A.~Said, Y.~H. Chao, and A.~Ortega, ``Graph-based transforms for
  inter predicted video coding,'' in \emph{Proc. IEEE International Conference
  on Image Processing (ICIP)}, 2015, pp. 3992--3996.

\bibitem{ioa}
G.~Fracastoro and E.~Magli, ``Predictive graph construction for image
  compression,'' in \emph{Proc. IEEE International Conference on Image
  Processing (ICIP)}, 2015, pp. 2204--2208.

\bibitem{iob}
G.~Fracastoro, F.~Verdoja, M.~Grangetto, and E.~Magli, ``Superpixel-driven
  graph transform for image compression,'' in \emph{Proc. IEEE International
  Conference on Image Processing (ICIP)}, 2015, pp. 2631--2635.

\bibitem{ioc}
G.~Fracastoro and E.~Magli, ``Steerable discrete cosine transform,'' in
  \emph{Proc. IEEE International Workshop on Multimedia Signal Processing, 2015
  (MMSP)}, 2015.

\bibitem{merris1994laplacian}
R.~Merris, ``Laplacian matrices of graphs: a survey,'' \emph{Linear algebra and
  its applications}, vol. 197, pp. 143--176, 1994.

\bibitem{chung1997spectral}
F.~R.~K. Chung, \emph{Spectral graph theory}.\hskip 1em plus 0.5em minus
  0.4em\relax American Mathematical Soc., 1997, vol.~92.

\bibitem{strang1999discrete}
G.~Strang, ``The discrete cosine transform,'' \emph{SIAM review}, vol.~41,
  no.~1, pp. 135--147, 1999.

\bibitem{zhang2013analyzing}
C.~Zhang and D.~Flor{\^e}ncio, ``Analyzing the optimality of predictive
  transform coding using graph-based models,'' \emph{IEEE Signal Process.
  Lett.}, vol.~20, no.~1, pp. 106--109, 2013.

\bibitem{merris1998laplacian}
R.~Merris, ``Laplacian graph eigenvectors,'' \emph{Linear algebra and its
  applications}, vol. 278, no.~1, pp. 221--236, 1998.

\bibitem{golub1996matrix}
G.~H. Golub and C.~F. Van~Loan, \emph{Matrix computations}.\hskip 1em plus
  0.5em minus 0.4em\relax JHU Press, 1996.

\bibitem{kim2001novel}
Y.~K. Kim, Z.~He, and S.~K. Mitra, ``A novel linear source model and a unified
  rate control algorithm for {H}. 263/{MPEG}-2/{MPEG}-4,'' in \emph{Proc. IEEE
  International Conference on Acoustics, Speech, and Signal Processing
  (ICASSP)}, 2001, pp. 1777--1780.

\bibitem{mal98}
S.~Mallat and F.~Falzon, ``Analysis of low bit rate image transform coding,''
  \emph{IEEE Trans. Signal Process.}, no.~4, pp. 1027--1042, 1998.

\bibitem{sullivan1998rate}
G.~J. Sullivan and T.~Wiegand, ``Rate-distortion optimization for video
  compression,'' \emph{IEEE Signal Process. Mag.}, vol.~15, no.~6, pp. 74--90,
  1998.

\bibitem{blu08}
T.~Blumensath and M.~E. Davies, ``Iterative thresholding for sparse
  approximations,'' \emph{J. Fourier Anal. Appl.}, vol.~14, no. 5-6, pp. 629 --
  654, 2008.

\bibitem{wang2004image}
Z.~Wang, A.~C. Bovik, H.~R. Sheikh, and E.~P. Simoncelli, ``Image quality
  assessment: from error visibility to structural similarity,'' \emph{IEEE
  Trans. Image Process.}, vol.~13, no.~4, pp. 600--612, 2004.

\bibitem{bjontegaard2001calcuation}
G.~Bjontegaard, ``Calculation of average {PSNR} differences between
  {RD}-curves,'' \emph{Doc. VCEG-M33 ITU-T Q6/16, Austin, TX, USA}, 2001.

\bibitem{img_db}
USC-SIPI, Image database, vol. 3: Miscellaneous,
  \url{sipi.usc.edu/database/database.php?volume=misc}.

\bibitem{taubman2012jpeg2000}
D.~Taubman and M.~Marcellin, \emph{JPEG2000 Image Compression Fundamentals,
  Standards and Practice: Image Compression Fundamentals, Standards and
  Practice}.\hskip 1em plus 0.5em minus 0.4em\relax Springer Science \&
  Business Media, 2012, vol. 642.

\bibitem{sikora2005trends}
T.~Sikora, ``Trends and perspectives in image and video coding,'' \emph{Proc.
  IEEE}, vol.~93, no.~1, pp. 6--17, 2005.

\bibitem{budagavi2013core}
M.~Budagavi, A.~Fuldseth, G.~Bj{\o}ntegaard, V.~Sze, and M.~Sadafale, ``Core
  transform design in the high efficiency video coding ({HEVC}) standard,''
  \emph{IEEE J. Sel. Topics Signal Process.}, vol.~7, no.~6, pp. 1029--1041,
  2013.

\bibitem{wien2015high}
M.~Wien, \emph{High Efficiency Video Coding}.\hskip 1em plus 0.5em minus
  0.4em\relax Springer, 2015.

\end{thebibliography}
\end{document}